\documentclass[sigconf]{acmart}

\usepackage{graphicx}
\usepackage{subcaption}
\usepackage{pifont}

\newcommand{\cmark}{\ding{51}}
\newcommand{\slfrac}[2]{\left.#1\middle/#2\right.}
\newtheorem{lemma}{Lemma}

\AtBeginDocument{%
  \providecommand\BibTeX{{%
    \normalfont B\kern-0.5em{\scshape i\kern-0.25em b}\kern-0.8em\TeX}}}

\copyrightyear{2022}
\acmYear{2022}
\setcopyright{acmcopyright}
\acmConference[SPAA '22]{Proceedings of the 34th ACM Symposium on Parallelism in Algorithms and Architectures}{July 11--14, 2022}{Philadelphia, PA, USA}
\acmBooktitle{Proceedings of the 34th ACM Symposium on Parallelism in Algorithms and Architectures (SPAA '22), July 11--14, 2022, Philadelphia, PA, USA}
\acmPrice{15.00}
\acmDOI{10.1145/3490148.3538585}
\acmISBN{978-1-4503-9146-7/22/07}

\begin{document}
\title{Scalable Fine-Grained Parallel Cycle Enumeration Algorithms}

\author{Jovan Blanu\v{s}a}
\affiliation{
  \institution{IBM Research Europe}
  \city{Zurich}
  \country{Switzerland}
}
\additionalaffiliation{
  \institution{Ecole Polytechnique Fédérale de Lausanne (EPFL)}
  \department{School of Computer and Communication Sciences}
  \city{CH-1015 Lausanne}
  \country{Switzerland}
}
\email{jov@zurich.ibm.com}

\author{Paolo Ienne}
\affiliation{
  \institution{Ecole Polytechnique Fédérale de Lausanne (EPFL)}
  \department{School of Computer and Communication Sciences}
  \city{CH-1015 Lausanne}
  \country{Switzerland}
}
\email{paolo.ienne@epfl.ch}

\author{Kubilay Atasu}
\affiliation{%
  \institution{IBM Research Europe}
  \city{Zurich}
  \country{Switzerland}
}
\email{kat@zurich.ibm.com}


\begin{abstract}
Enumerating simple cycles has important applications in computational biology, network science, and financial crime analysis.
In this work, we focus on parallelising the state-of-the-art simple cycle enumeration algorithms by Johnson and Read-Tarjan along with their applications to temporal graphs.
To our knowledge, we are the first ones to parallelise these two algorithms in a fine-grained manner. We are also the first to demonstrate experimentally a linear performance scaling.
Such a scaling is made possible by our decomposition of long sequential searches into fine-grained tasks, which are then dynamically scheduled across CPU cores, enabling an optimal load balancing.
Furthermore, we show that coarse-grained parallel versions of the Johnson and the Read-Tarjan algorithms that exploit edge- or vertex-level parallelism are not scalable.
On a cluster of four multi-core CPUs with $256$ physical cores, our fine-grained parallel algorithms are, on average, an order of magnitude faster than their coarse-grained parallel counterparts.
The performance gap between the fine-grained and the coarse-grained parallel algorithms widens as we use more CPU cores.
When using all 256 CPU cores, our parallel algorithms enumerate temporal cycles, on average, $260\times$ faster than the serial algorithm of Kumar and Calders.
\\

\vspace{-.07in}
\noindent \textbf{Code repository}:{ \small https://github.com/IBM/parallel-cycle-enumeration}
\end{abstract}


\maketitle

\section{Introduction}

\begin{figure}[t]
    \includegraphics[width=0.98\linewidth]{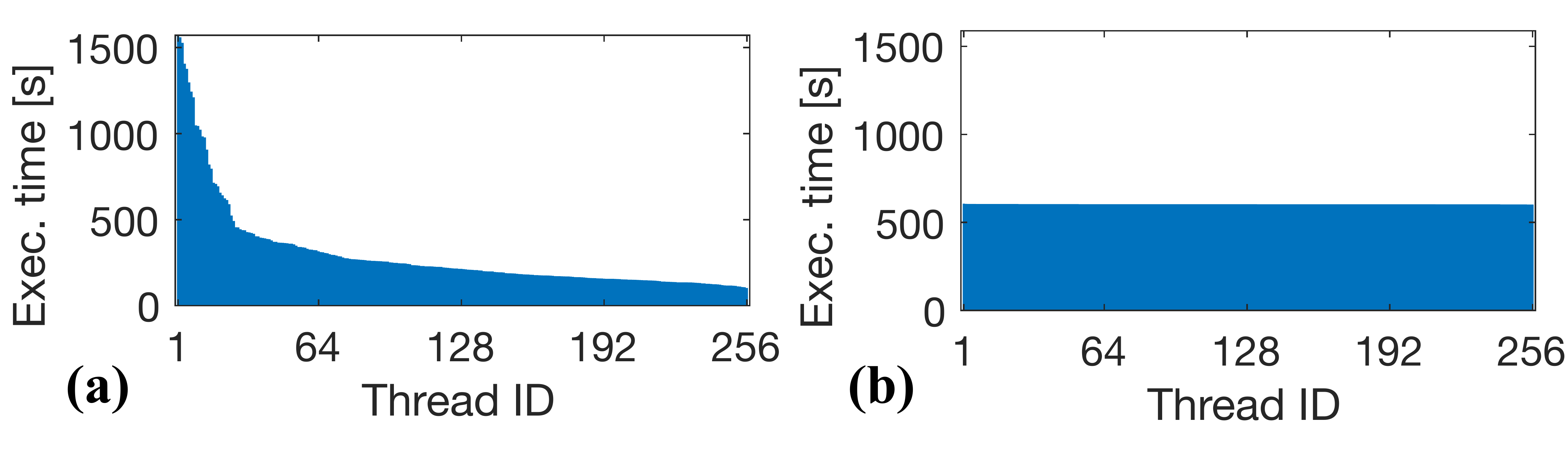}
	\vspace{-.15in}
	\caption{Per-thread execution time of (a) the coarse-grained Johnson algorithm vs. (b) our fine-grained parallel algorithm
	using the \textit{wiki-talk} graph and a $\mathbf{12h}$ time window.
	Thanks to a perfect load balancing, our fine-grained method is $\mathbf{3\times}$ faster on a $\mathbf{64}$-core CPU executing $\mathbf{256}$ threads.}
	\vspace{-.15in}
	\label{fig:motivfig}
\end{figure}

A graph-based data representation is desirable when analyzing large and complex datasets because it exposes the connectivity of the underlying data objects and enables the discovery of complex relationships between them~\cite{needham2019graph}.
Analysing graph-structured data has important applications in many domains, such as finance~\cite{neo4j_whitepaper_finance}, healthcare~\cite{wang_recent_2020}, cybersecurity~\cite{noel_cygraph_2016}, and advertising~\cite{neo4j_whitepaper_reccomendation}.
The existence of certain patterns, such as cycles, cliques, and motifs, in a graph can reveal nontrivial relationships between different graph objects~\cite{aggarwal_managing_2010}.
As the volume of graph data continues to grow, the discovery of such relationships becomes computationally challenging, necessitating more efficient algorithms and scalable parallel implementations that can exploit modern multi-core processors.

\textbf{Simple cycles and temporal cycles.} This paper introduces efficient parallel algorithms for enumerating simple cycles of directed graphs.
A simple cycle is a sequence of edges that starts and ends with the same vertex and visits other vertices of the graph at most once. 
Enumerating simple cycles has important applications in computational biology~\cite{kwon_analysis_2007, klamt_computing_2009}, network science~\cite{giscard_evaluating_2017,zhou_cycle_2018}, software bug tracking~\cite{sas_example}, and electronic design automation~\cite{neiroukh_transforming_2008, gupta_acyclic_2005, pothukuchi_dhuria_2021}.

Furthermore, some graphs have their edges annotated with timestamps, which we refer to as temporal graphs. In such graphs one can also look for temporal cycles~\cite{kumar_2scent_2018}, which are special cases of simple cycles, in which the edges are ordered in time.
For instance, in financial transaction graphs, a temporal cycle represents a series of transactions in which the money initially sent from one bank account returns back to the same account; the existence of such cycles is a strong indicator of financial fraud such as money laundering, tax avoidance~\cite{hajdu_temporal_2020, AMLSim}, and credit card fraud~\cite{qiu_real-time_2018}.
Finding temporal cycles in temporal graphs also enables detecting circular trading, which can be used for manipulating stock prices~\cite{palshikar_collusion_2008, islam_approach_2009, jiang_trading_2013}.

\textbf{Parallelisation challenges.} We focus on parallelising the algorithms by Johnson~\cite{johnson_finding_1975} and Read-Tarjan~\cite{read_bounds_1975} for finding cycles because these algorithms achieve the lowest time complexity bounds reported for directed graphs~\cite{mateti_algorithms_1976}.
Both algorithms are recursively formulated and construct a recursion tree in a depth-first fashion.
However, these algorithms employ different pruning techniques to limit the size of their recursion trees.
In practice, the Johnson algorithm is faster than the Read-Tarjan algorithm because it uses more aggressive recursion-tree pruning techniques~\cite{kao_enumeration_2016,mateti_algorithms_1976}.

The na\"ive way of parallelising the aforementioned algorithms involves searching for cycles starting from different vertices or edges in parallel, which we refer to as the coarse-grained parallel methods.
Such coarse-grained parallel approaches are straightforward to implement using the popular vertex-centric~\cite{malewicz_pregel_2010, mccune_thinking_2015} and edge-centric~\cite{roy_x-stream_2013} graph processing frameworks.
However, real-world graphs often exhibit a power-law or a log-normal distribution of vertex degrees~\cite{barabasi_network_2016, broido_scale-free_2019}.
In such graphs, the execution time of coarse-grained parallel approaches is dominated by searches that start from a small set of vertices or edges as illustrated in Figure~\ref{fig:motivfig}a.
This behaviour leads to a workload imbalance and limits scalability.

The shortcomings of coarse-grained parallel approaches can be addressed by decomposing the search for cycles starting from a given edge or vertex into finer-grained tasks.
Fine-grained parallelism has been exploited by other graph mining algorithms~\cite{blanusa_manycore_2020, das_shared-memory_2019, abdelhamid_scalemine_2016}.
However, to our knowledge, it has not been applied to asymptotically-optimal simple cycle enumeration algorithms, such as the Johnson algorithm and the Read-Tarjan algorithm.
In particular, the pruning efficiency of the Johnson algorithm depends on a strictly sequential depth-first-search-based recursion tree exploration. As such, decomposing the Johnson algorithm into fine-grained tasks is not possible without giving up some of its pruning efficiency. In contrast, the Read-Tarjan algorithm does not require a strictly sequential depth-first-search-based recursion tree exploration, hence, it is easier to decompose into fine-grained tasks.

\begin{table}[t]
\centering
\caption{
Our fine-grained parallel Read-Tarjan algorithm is the only solution that is both work efficient and scalable.}
\vspace{-.1in}
\addtolength{\tabcolsep}{-3pt}
\begin{tabular}{l|cc}
\textbf{Parallel algorithm}               &  \textbf{Work efficient} & \textbf{Scalable}  \\ \hline
Coarse-grained parallel algorithms  &  \cmark  &        \\
Our fine-grained parallel Johnson    &  & \cmark   \\
Our fine-grained parallel Read-Tarjan  & \cmark  & \cmark     \\ \hline
\end{tabular}
\label{tab:theoSummary}
\vspace{-.2in}
\end{table}

\textbf{Contributions.} 
In this paper, we contribute scalable parallel versions of the Johnson and the Read-Tarjan algorithms.
To our knowledge, we are the first ones to parallelise these two algorithms in a fine-grained manner. We are also the first to demonstrate an almost linear performance scaling on a system that can execute up to a thousand concurrent software threads.
Such a scalability is enabled by our decomposition of long sequential searches into fine-grained tasks, which are then dynamically scheduled across CPU cores, leading to an ideal load balancing as shown in Figure~\ref{fig:motivfig}b.

To decompose the Johnson algorithm into fine-grained tasks, we have relaxed its strictly depth-first-search-based exploration. In this way, we have enabled it to perform multiple independent depth-first searches in parallel. However, this additional flexibility is at the expense of some pruning efficiency. Because of the reduced pruning efficiency, our fine-grained parallel Johnson algorithm performs more work than its serial version--- i.e., it is not work-efficient. In contrast, our fine-grained parallel Read-Tarjan algorithm does not perform more work than its sequential version and is work efficient. Nevertheless, both fine-grained algorithms are scalable both in theory and in practice.
Table~\ref{tab:theoSummary} shows the key results of our theoretical analysis. Interestingly, despite not being work efficient, our fine-grained Johnson algorithm  outperforms our fine-grained parallel Read-Tarjan algorithm in most of our experiments.

\textbf{Paper structure.} The remainder of this paper is organised as follows.
Section~\ref{sect:related_work} discusses the related work. 
Section~\ref{section:background} introduces the notation used, formally defines the concepts of work efficiency and scalability, and covers state-of-the-art sequential algorithms for finding simple cycles.
Section~\ref{sect:vertEdgePar} covers coarse-grained parallel versions of the Johnson and the Read-Tarjan algorithms.
Section~\ref{sect:tpJohnson} and Section~\ref{sect:tpReadTarjan} introduce our fine-grained parallel versions of the Johnson and the Read-Tarjan algorithms, respectively.
Section~\ref{sect:temporalCycle} discusses adaptations of our algorithms to compute temporal cycles.
Section~\ref{sect:experiments} provides an experimental evaluation of all the parallel algorithms covered in this work.
Section~\ref{sect:conclusion} presents our conclusions.

\section{Related Work}
\label{sect:related_work}
\textbf{Simple cycle enumeration algorithms.}
Enumeration of simple cycles of graphs is a classical computer science problem~\cite{tiernan_efficient_1970, tarjan_enumeration_1973, johnson_finding_1975, read_bounds_1975, mateti_algorithms_1976, Szwarcfiter1976ASS, kao_enumeration_2016, weinblatt_new_1972, loizou_enumerating_1982, birmele_optimal_2013}. 
The backtracking-based algorithms by Johnson~\cite{johnson_finding_1975}, Read and Tarjan~\cite{read_bounds_1975}, and Szwarcfiter and Lauer~\cite{Szwarcfiter1976ASS} achieve the lowest time complexity bounds reported in the literature for enumerating simple cycles in directed graphs.
These algorithms implement advanced recursion tree pruning techniques to improve on the brute-force  Tiernan algorithm~\cite{tiernan_efficient_1970}. Section~\ref{section:back_johnson} covers such pruning techniques in further detail.
A cycle enumeration algorithm that is asymptotically faster than the aforementioned algorithms~\cite{johnson_finding_1975, read_bounds_1975, Szwarcfiter1976ASS} has been proposed in Birmelé et al.~\cite{birmele_optimal_2013}, however, it is applicable only to undirected graphs.
The algorithms for simple cycle enumeration can be specialised to find temporal cycles, such as in Kumar and Calders~\cite{kumar_2scent_2018}, and our parallel algorithms lend themselves to the same specialisation.
Simple cycles can also be enumerated by computing the powers of the adjacency matrix~\cite{danielson_finding_1968, kamae_systematic_1967, ponstein_self-avoiding_1966} or by using circuit vector space algorithms~\cite{mateti_algorithms_1976, gibbs_cycle_1969, welch_numerical_1965}, but the complexity of such approaches grows exponentially with the length of the cycles or the size of the graphs.

\begin{table}[t]
\centering
\caption{Capabilities of the related work versus our own. Competing algorithms either fail to exploit fine-grained parallelism or do it on top of asymptotically inferior \mbox{algorithms}.}
\vspace{-.15in}
\addtolength{\tabcolsep}{-2pt}
\begin{tabular}{l|cccccc}
\textbf{Related work} & \textbf{\cite{kumar_2scent_2018}} & \textbf{\cite{qiu_real-time_2018}} & \textbf{\cite{peng_towards_2019}} & \textbf{\cite{nah_efficient_2020}} & \textbf{\cite{gupta_finding_2021}}  & \textbf{Ours} \\ 
\hline 
Fine-grained parallelism    &       &       &       & \cmark &       & \cmark \\
Asymptotic optimality       & \cmark &       & \cmark &       & \cmark & \cmark \\ 
Temporal cycles             & \cmark &       &       &       &       & \cmark \\
Time-window constraints     & \cmark & \cmark &       &       &       & \cmark \\
Cycle-length constraints    &       & \cmark & \cmark & \cmark & \cmark & \\
\hline
\end{tabular}
\label{tab:relWork}
\vspace{-.15in}
\end{table}

\textbf{Cycle-length and time-window constraints.}
To make cycle enumeration problem tractable, it is common to search for cycles under some constraints. For instance, the length of the cycles---i.e., the maximum number of edges in the cycle, can be constrained, such as in Gupta and Suzumura~\cite{gupta_finding_2021}, Peng et al.~\cite{peng_towards_2019}, and Qiu et al~\cite{qiu_real-time_2018}. In temporal graphs, it is also common to search for cycles within a sliding time window, such as in Kumar and Calders~\cite{kumar_2scent_2018} and Qiu et al~\cite{qiu_real-time_2018}.
Constraining the length of the cycles or the size of the time windows effectively narrows down the search to a spatial or temporal neighbourhood, respectively.
In this work, we focus on temporal graphs, therefore, we use time window constraints when enumerating both simple and temporal cycles.
Note that length-constrained simple cycles can also be enumerated using incremental algorithms, such as in Qiu et al.~\cite{qiu_real-time_2018}.
However, this algorithm is based on the brute-force Tiernan algorithm~\cite{tiernan_efficient_1970}, which makes it slower than nonincremental algorithms that use recursion tree pruning techniques~\cite{peng_towards_2019}.
In addition, because incremental algorithms maintain auxiliary data structures, such as paths, to be able to construct cycles incrementally, they are not as memory-efficient as nonincremental algorithms~\cite{peng_towards_2019}.
Table~\ref{tab:relWork} offers comparisons between the capabilities of these methods and ours.

\textbf{Parallel and distributed algorithms for cycle enumeration.}
Cui et al. \cite{cui_multi-threading_2017} proposed a multi-threaded algorithm for detecting and removing simple cycles of a directed graph.
The algorithm divides the graph into its strongly-connected components and each thread performs a depth-first search on a different component to find cycles.
However, sizes of the strongly-connected components in real-world graphs can vary significantly~\cite{meusel_graph_2014}, which leads to a workload imbalance.
Rocha and Thatte~\cite{rocha_distributed_2015} proposed a distributed algorithm for simple cycle enumeration based on the bulk-synchronous parallel model~\cite{valiant_bridging_1990}, but it searches for cycles in a brute-force manner.
Qing et al.~\cite{nah_efficient_2020} introduced a parallel algorithm for finding length-constrained simple cycles. 
It is the only other fine-grained parallel algorithm we are aware of in the sense that it can search for cycles starting from the same vertex in parallel.
However, the way this algorithm searches for cycles is similar to the way the brute-force Tiernan algorithm~\cite{tiernan_efficient_1970} works.
To our knowledge, we are the first ones to introduce fine-grained parallel versions of asymptotically-optimal simple cycle enumeration algorithms, which do not rely on a brute-force search, as we show in Table~\ref{tab:relWork}.

\section{Background}
\label{section:background}

This section introduces the main theoretical concepts used in this paper and provides an overview of the most prominent simple cycle enumeration algorithms.
The notation used is given in Table~\ref{tab:notation}.

\subsection{Preliminaries}
\label{sect:back_prelim}

\begin{table}[t]
\centering
\caption{Summary of the notation used in the paper.
\vspace{-.1in}
}
\addtolength{\tabcolsep}{-2pt}
\begin{tabular}{l|l}
	\textbf{Symbol}            &    \textbf{Description}       \\ \hline
	\textbf{$u \rightarrow v$} & A directed edge connecting vertex $u$ with $v$. \\
    \textbf{$n$, $e$}    & Number of vertices and edges in a graph.    \\
    \textbf{$\delta$} & Size of a time window. \\
	\textbf{$c$}       & Number of simple cycles in a graph. \\
    \textbf{$s$} & Number of maximal simple paths in a graph. \\
    \textbf{$\Pi$}  & The current simple path explored by an algorithm. \\
    \textbf{$\mathit{Blk}$} & The set of blocked vertices. \\
    \textbf{$\mathit{Blist}$} & The unblock list data structure. \\
	\textbf{$p$}       &  Number of threads used for parallel algorithms. \\
	\textbf{$T_{p}(n)$}       & Execution time of a parallel algorithm. \\ 
	\textbf{$W_{p}(n)$}       & Amount of work a parallel algorithm performs. \\ 
	\hline
\end{tabular}
\label{tab:notation}
\vspace{-.1in}
\end{table}

We consider a directed graph $\mathcal{G}(\mathcal{V}, \mathcal{E})$ having a set of vertices $\mathcal{V}$ and a set of directed edges $\mathcal{E} = \{ u \rightarrow v \mid u, v \in \mathcal{V}\}$.
The set of neighbors of a given vertex $v$ is defined as $\mathcal{N}(v) = \{ w \mid \forall \; v \rightarrow w \in \mathcal{E}\}$.
An outgoing edge of a given vertex $v$ is defined as $v \rightarrow w$ and an incoming edge is defined as $u \rightarrow v$, where $v \rightarrow w, u \rightarrow v  \in \mathcal{E}$.
A \emph{path} between the vertices $v_0$ and $v_k$, denoted as $v_0 \rightarrow v_1 \ldots \rightarrow v_k$, is a sequence of vertices such that there exists an edge between two consecutive vertices of the sequence.
A \emph{simple path} is a path with no repeated vertices.
A simple path is \emph{maximal} if the last vertex of the path has no neighbors or all of its neighbors are already in the path~\cite{erdos_maximal_1959}.
A \textbf{cycle} is a path of non-zero length from a vertex $v$ to the same vertex $v$.
A \textbf{simple cycle} is a cycle with no repeated vertices except for the first and the last vertex.
The number of maximal simple paths and the number of simple cycles in a graph are denoted as $s$ and $c$, respectively (see Table~\ref{tab:notation}).
Note that $s$ can be exponentially larger than $c$~\cite{tarjan_enumeration_1973}.
The goal of \textbf{simple cycle enumeration} is to compute all simple cycles of a directed graph $\mathcal{G}$, ideally without computing all maximal simple paths of it.

A \textbf{temporal graph} is a graph that has its edges annotated with timestamps.~\cite{paranjape_motifs_2017}.
In temporal graphs, a \textbf{temporal cycle} is a simple cycle, in which the edges appear in the increasing order of their timestamps.
A simple cycle or a temporal cycle of a temporal graph occurs within a \textbf{time window} $\left[t_{w1}: t_{w2} \right]$ if every edge of that cycle has a timestamp $t_s$ such that $t_{w1} \leq t_s \leq t_{w2}$.
Figure~\ref{fig:time-window} shows the simple cycles of a temporal graph that occur within two different time windows of size $\delta = 5$.
This graph contains one simple cycle in the time window $\left[2: 7\right]$ (Figure~\ref{fig:time-window-a}), which is also a temporal cycle, and two simple cycles in the time window $\left[10: 15\right]$ (Figure~\ref{fig:time-window-b}).

\begin{figure}[t]
	\begin{subfigure}[t]{0.5\linewidth}
		\centering
		\includegraphics[width=0.8\linewidth]{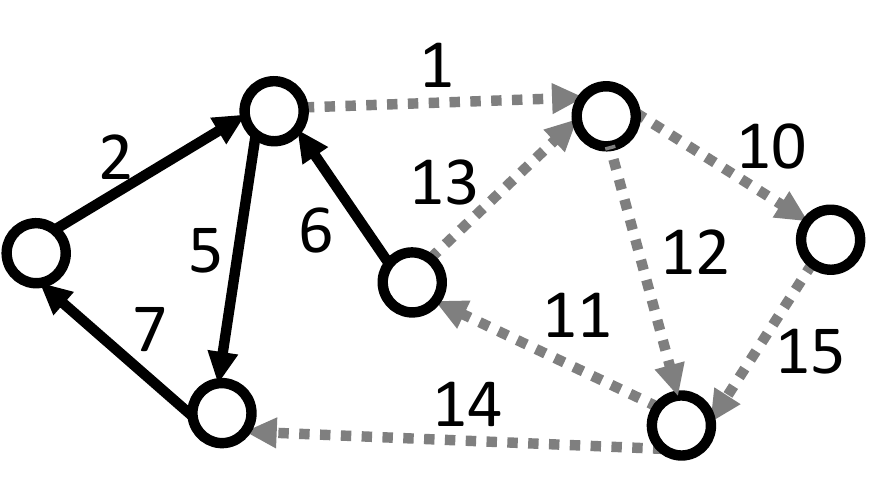}
	\vspace{-.05in}
		\caption{Time window $\left[2:7\right]$}
		\label{fig:time-window-a}
	\end{subfigure}%
	\begin{subfigure}[t]{0.5\linewidth}
		\centering
		\includegraphics[width=0.8\linewidth]{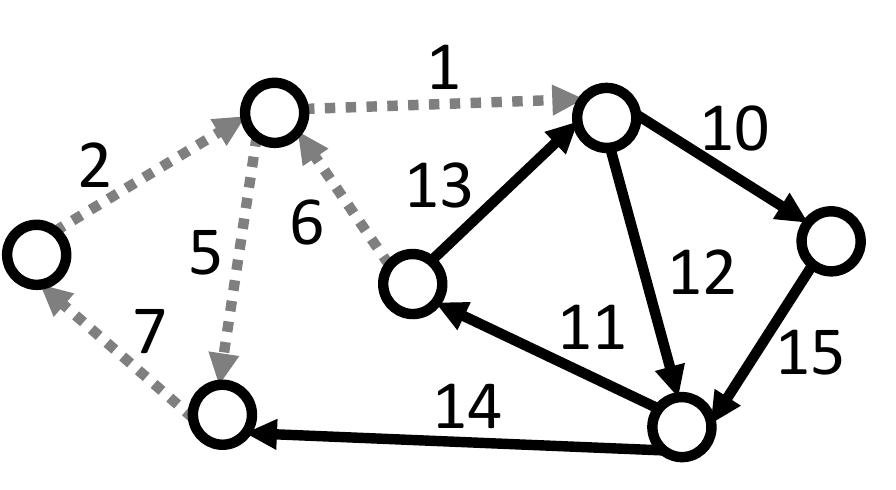}
	\vspace{-.05in}
		\caption{Time window $\left[10:15\right]$}
		\label{fig:time-window-b}
	\end{subfigure}
	\vspace{-.15in}
	\caption{Two snapshots of a temporal graph associated with two different time windows of size $\mathbf{\delta = 5}$. The solid arrows indicate the edges that belong to the respective time windows.}
	\vspace{-.15in}
	\label{fig:time-window}
\end{figure}

\subsection{Task-level parallelism}
The parallel algorithms described in this paper can be implemented using shared-memory parallel processing frameworks, such as TBB~\cite{kukanov_foundations_2007}, Cilk~\cite{blumofe_cilk_1996}, and OpenMP~\cite{quinn_parallel_2004}.
These frameworks enable decomposition of a program into tasks that can be independently executed by different software threads.
In our setup, tasks are dynamically created and scheduled. A \emph{parent} task can create several \emph{child} tasks.
A dynamic task management system assigns the tasks created to the work queues of available threads.
Furthermore, a work-stealing scheduler~\cite{blumofe_scheduling_1999, kukanov_foundations_2007, blumofe_cilk_1996} enables a thread that is not executing a task to \emph{steal} a task from the work queue of another thread.
Stealing tasks enables dynamic load balancing and ensures full utilisation of the threads when there are sufficiently many tasks.

\subsection{Work efficiency and scalability}
\label{sect:back_eff}

We use the notions of \emph{work efficiency} and \emph{scalability} to analyse parallel algorithms~\cite{par_algos}.
We refer to the time to execute a parallel algorithm on a problem of size $n$ using $p$ threads as $T_p(n)$. The size of a graph is determined by  the number of vertices $n$ as well as the number of edges $e$, but we will refer only to $n$ for simplicity.
The \emph{depth} of an algorithm is the length of the longest sequence of dependent operations in the algorithm. The time it takes to execute such a sequence is equal to the execution time of the parallel algorithm using an infinite number of threads, denoted by $T_{\infty}$.
In addition, the \emph{work} performed by a parallel algorithm that uses $p$ threads is the sum of the execution times of the individual threads.
The \emph{work efficiency} and the \emph{scalability} are formally defined as follows.

\begin{definition}
\label{def:workEfficiency}
(\textit{Work efficiency}) 
A parallel algorithm is work efficient if and only if $W_p(n) \in O(T_1(n))$.
\end{definition}

\begin{definition}
\label{def:scalability}
(\textit{Scalability}) 
A parallel algorithm is scalable if and only if $\lim\limits_{n\to\infty} \left(\lim\limits_{p\to\infty} \dfrac{T_p(n)}{T_1(n)} \right) = 0$.
\end{definition}

Informally, a work efficient parallel algorithm performs no more work than its serial version.
Moreover, scalability implies that, for large enough inputs, increasing the number of threads increases the speedup of the parallel algorithm with respect to its serial version.

We also define the notion of \emph{strong scalability} as follows~\cite{JaJa1992-va}.

\begin{definition}
\label{def:strongScalability}
(\textit{Strong scalability}) 
A parallel algorithm is strongly scalable if and only if $\dfrac{T_1(n)}{T_p(n)} = \Theta(p)$ for large enough $n$.
\vspace{-.03in}
\end{definition}

Whereas Definition~\ref{def:scalability} implies that the speedup $T_1(n)/T_p(n)$ achieved by a parallel algorithm with respect to its serial execution is infinite when the number of threads $p$ is infinite, Definition~\ref{def:strongScalability} implies that the speedup is always in the order of $p$.
Another related concept is weak scalability, which requires the speedup to be in the order of $p$ when the input size per thread is constant.
Note that both strong scalability and weak scalability  guarantee scalability.

\subsection{Simple cycle enumeration algorithms}

The following algorithms for simple cycle enumeration perform recursive searches to incrementally update simple paths that can lead to cycles.
Each algorithm iterates the vertices or edges of the graph and independently constructs a recursion tree to enumerate all the cycles starting from that vertex or edge.
The difference between these algorithms is to what extent they reduce the redundant work performed during the recursive search, which we discuss next.

\label{section:back_tiernan}

\textbf{The Tiernan algorithm}~\cite{tiernan_efficient_1970} enumerates simple cycles using a brute-force search.
It recursively extends a simple path $\Pi$ by appending a neighbor $u$ of the last vertex $v$ of $\Pi$ provided that  $u$ is not already in $\Pi$.
A clear downside of this algorithm is that it can repeatedly visit vertices that can never lead to a cycle.
When searching for cycles in the graph shown in Figure~\ref{fig:background}a starting from vertex $v_0$, this algorithm would explore the path $b_{1}, \ldots, b_{k}$ $2m$ times.
From each vertex $w_i$ and $u_i$, with $i \in \{1, \ldots, m\}$, the Tiernan algorithm would explore this path only to discover that it cannot lead to a simple cycle. 
As noted by Tarjan~\cite{tarjan_enumeration_1973}, the Tiernan algorithm explores every simple path and, consequently, all maximal simple paths of a graph.
Exploring a maximal simple path takes $O(n+e)$ time because a path can contain up to $n$ vertices, and the Tiernan algorithm explores every outgoing edge of every vertex in that path.
Given a graph with $s$ maximal simple paths (see Table~\ref{tab:notation}), the worst-case time complexity of the Tiernan algorithm is $O(s(n+e))$.

\label{section:back_johnson}

\textbf{The Johnson algorithm}~\cite{johnson_finding_1975} improves upon the Tiernan algorithm by avoiding the vertices that cannot lead to simple cycles.
For this purpose, the Johnson algorithm maintains a set of blocked vertices $\mathit{Blk}$ that are avoided during the search.
In addition, a list of vertices $\mathit{Blist}[w]$ is stored for each vertex $w$.
Whenever a vertex $w$ is unblocked (i.e., removed from $\mathit{Blk}$) by the Johnson algorithm, the vertices in $\mathit{Blist}[w]$ are also unblocked. 
This unblocking process is performed recursively until no more vertices can be unblocked, which we refer to as the \emph{recursive unblocking} procedure.

A vertex $v$ is blocked (i.e. added to $\mathit{Blk}$) when visited by the algorithm.
If a cycle is found after recursively exploring every neighbor of $v$ that is not blocked, the vertex $v$ is unblocked upon backtracking.
Otherwise, if no cycles are found by exploring the neighbors of $v$, $v$ is not unblocked immediately upon backtracking. The $\mathit{Blist}$ data structure is updated to enable unblocking of $v$ in a later step by adding $v$ to the list $\mathit{Blist}[w]$ of every neighbor $w$ of $v$.
This delayed unblocking of vertices enables the Johnson algorithm to discover each cycle in $O(n+e)$ time in the worst case~\cite{johnson_finding_1975}.
Because this algorithm also requires $O(n+e)$ time to determine that there are no cycles, its worst-case time complexity is $O\left((n+e)(c+1)\right)$.
Note that because $s$ can be exponentially larger than $c$~\cite{tarjan_enumeration_1973}, the Johnson algorithm is asymptotically faster than the Tiernan algorithm.

In the example shown in Figure~\ref{fig:background}a, every simple path $\Pi$ that starts from $v_0$ and goes through $b_1,\ldots,b_k$ vertices is a maximal simple path, and thus, it cannot lead to a simple cycle. 
The Johnson algorithm would block $b_1,\ldots,b_k$ immediately after visiting this sequence once and then keep these vertices blocked until it backtracks to $v_1$, at which point, the algorithm would have finished exploring both subtrees shown in Figure~\ref{fig:background}b.
As a result, the Johnson algorithm visits $b_1,\ldots,b_k$ vertices only once, rather than $2m$ times the Tiernan algorithm would visit them.
Note that because these vertices get blocked during the exploration of the left subtree of the recursion tree, they are not going to be visited again during the exploration of the right subtree. Effectively, a portion of the right subtree is pruned (see the dotted path in Figure~\ref{fig:background}b) based on the updates made on $\mathit{Blist}$ during the exploration of the left subtree.
This strictly sequential depth-first exploration of the recursion tree is critically important for achieving a high pruning efficiency, but it also makes scalable parallelisation of the Johnson algorithm extremely challenging, which
we are going to cover in Section~\ref{sect:tpJohnson}.

\begin{figure}[t]
	\centerline{
		\includegraphics[width=0.96\linewidth]{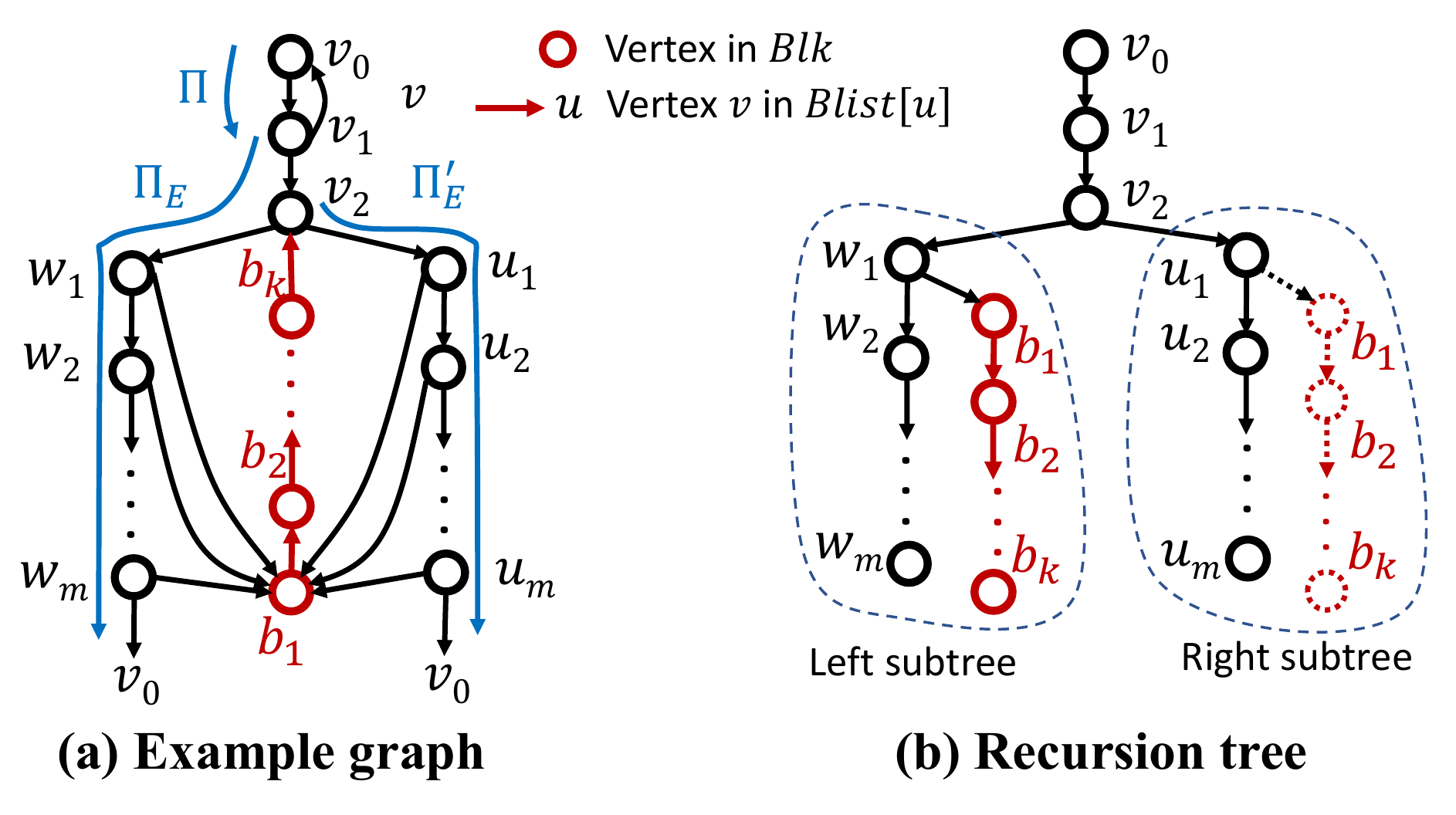}
	}
	\vspace{-.15in}
	\caption{
	(a) An example graph, and (b) the recursion tree constructed when searching for cycles in (a) starting from vertex $\mathbf{v_0}$.
	The nodes of the recursion tree represent the recursive calls of the depth-first search.
	The dotted path of the right subtree is explored only by the Read-Tarjan algorithm. 
	}
	\label{fig:background}
	\vspace{-.15in}
\end{figure}

\label{sect:read_tarjan}

\textbf{The Read-Tarjan algorithm}~\cite{read_bounds_1975} also has a worst-case time complexity of $O\left((n+e)(c+1)\right)$.
This algorithm maintains a current path $\Pi$ between a starting vertex and a frontier vertex. 
A recursive call of this algorithm iterates the neighbors of the current frontier vertex and performs a depth-first search (DFS). Assume that $v_0$ is the starting vertex and $v_1$ is the frontier vertex of $\Pi$ (see Figure~\ref{fig:background}a). From each neighbor $y \in \{v_0, v_2\}$ of $v_1$, a DFS tries to find a path extension $\Pi_E$ back to $v_0$ that would form a simple cycle when appended to $\Pi$.
In the example shown in Figure~\ref{fig:background}a, the algorithm finds two path extensions, one indicated as $\Pi_E$ and one that consists of the edge $v_1 \rightarrow v_0$.
The algorithm then explores each path extension by iteratively appending the vertices from it to the path $\Pi$.  
For each vertex $x$ of a path extension added to $\Pi$, the algorithm also searches for an alternate path extension from that vertex $x$ to $v_0$ using a DFS.
In the example given in Figure~\ref{fig:background}a, the algorithm iterates through the vertices of the path extension $\Pi_E$ and finds an alternate path extension $\Pi_{E}^{\prime}$ from the neighbor $u_1$ of $v_2$.
If an alternate path extension is found, a child recursive call is invoked with the updated current path $\Pi$, which is $v_0 \rightarrow v_1 \rightarrow v_2$ in our example.
Otherwise, if all the vertices in $\Pi_E$ have already been added to the current path $\Pi$, $\Pi$ is reported as a simple cycle. 
In our example, the Read-Tarjan algorithm explores both $\Pi_E$ and $\Pi_{E}^{\prime}$ path extensions, and each one leads to the discovery of a cycle.

The Read-Tarjan algorithm also maintains a set of blocked vertices $\mathit{Blk}$ for recursion-tree pruning. However, differently from the Johnson algorithm, $\mathit{Blk}$ only keeps track of the vertices that cannot lead to new cycles when exploring the current path extension.
The vertices in $\mathit{Blk}$ are avoided while searching for additional path extensions that branch from the current path extension.
For instance, the left subtree of the recursion tree shown in Figure~\ref{fig:background}b demonstrates the exploration of the path extension $\Pi_E$ shown in Figure~\ref{fig:background}a.
During the exploration of $\Pi_E$, vertices $b_1, \ldots, b_k$ are added to $\mathit{Blk}$ immediately after visiting $w_1$, and they are not visited again while exploring $\Pi_E$.
However, when exploring another path extension $\Pi_{E}^{\prime}$ in the right subtree, the vertices $b_1, \ldots, b_k$ are visited once again (see the dotted path of the right subtree).
As a result, the Read-Tarjan algorithm visits $b_1, \ldots, b_k$ twice instead of just once.
As we are going to show in Section~\ref{sect:tpReadTarjan}, this drawback becomes an advantage when parallelising the Read-Tarjan algorithm because it enables independent exploration of different subtrees of the recursion tree.

\textbf{Time window constraints} can be supported trivially by all three algorithms covered.
Such constraints restrict the search for simple cycles to those that occur within a time window of a given size $\delta$.
When a search for cycles starting from an edge with a timestamp $t$ is invoked, these algorithms consider only the edges with timestamps that belong to the time window $\left[t: t + \delta \right]$.
In consequence, fewer vertices are visited during the search for cycles.

\section{Coarse-grained parallel methods}
\label{sect:vertEdgePar}

The most straightforward way of parallelising the Johnson and the Read-Tarjan algorithms is to search for cycles that start from different vertices or edges in parallel.
Each such search can then execute on a different thread and construct its own recursion tree.
Such a coarse-grained approach to parallelising the cycle enumeration algorithms is work efficient.
However, it is not scalable, which we prove in this section.

\begin{figure}[t]
	\begin{subfigure}[t]{0.32\linewidth}
		\centering
		\includegraphics[height=3.7cm]{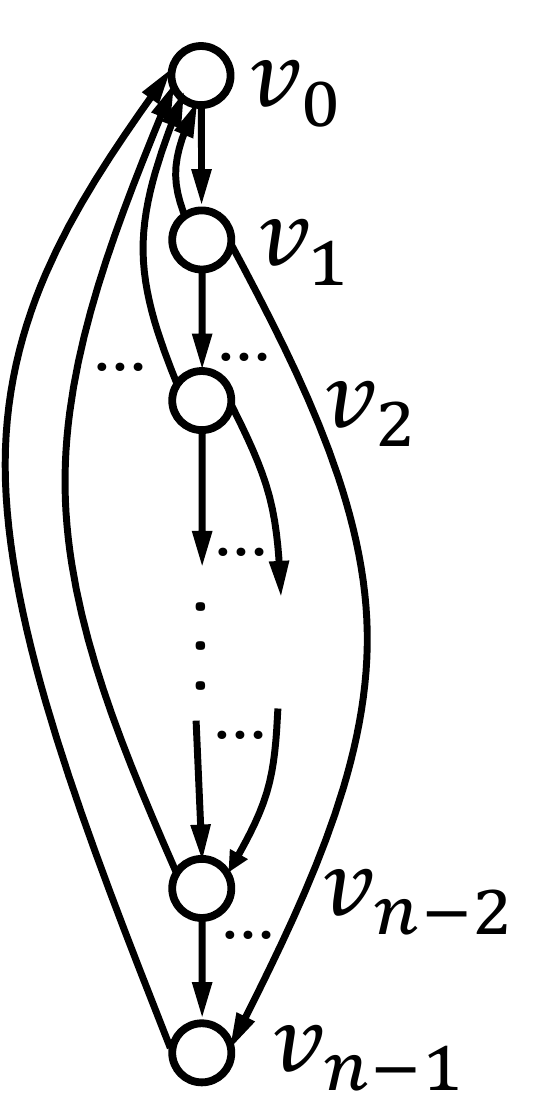}
	\vspace{-.05in}
		\caption{Example graph}
		\label{fig:wcEdgePar}
	\end{subfigure}%
	\begin{subfigure}[t]{0.58\linewidth}
		\centering
		\includegraphics[height=3.7cm]{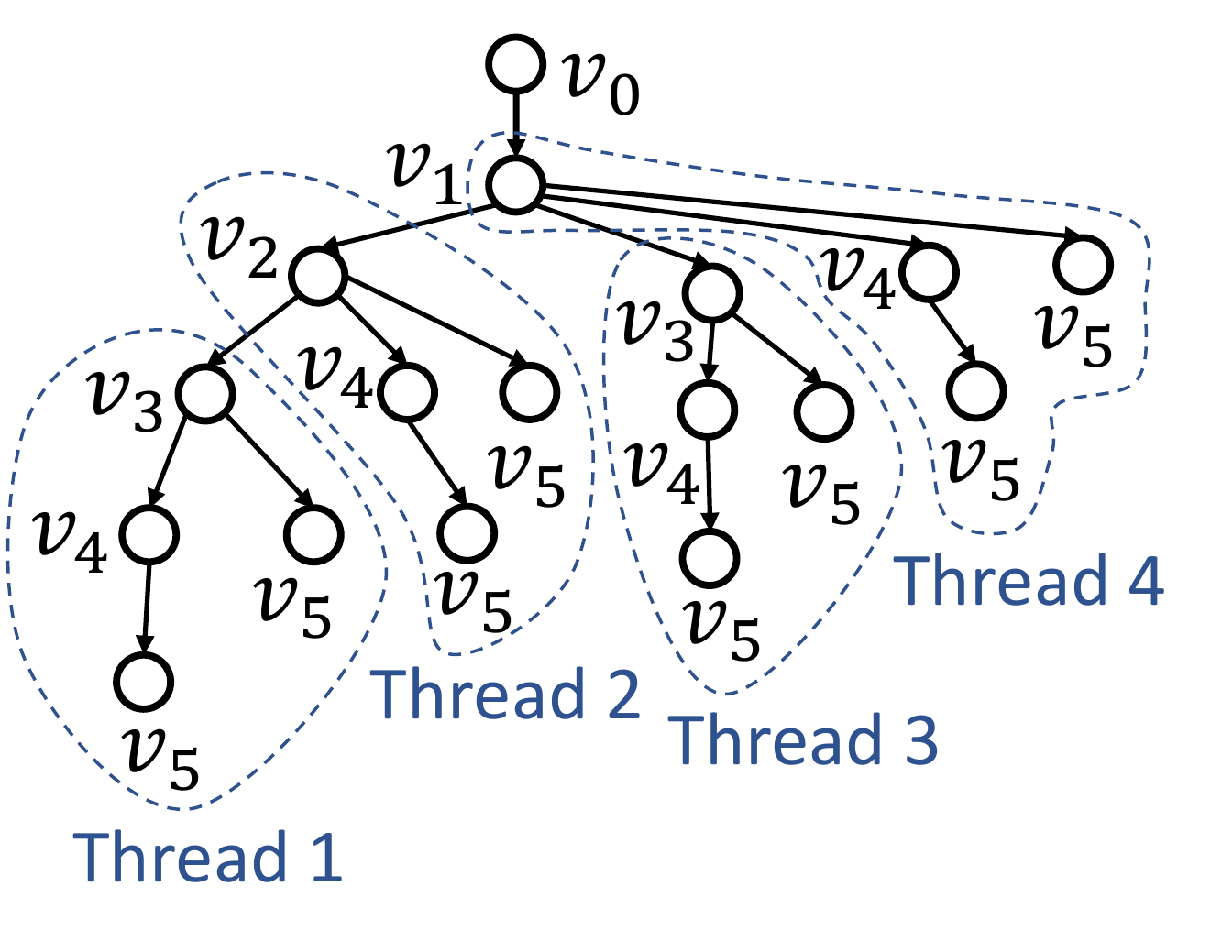}
	\vspace{-.05in}
		\caption{Recursion tree}
		\label{fig:wc_exampleBT}
	\end{subfigure}
	\vspace{-.1in}
	\caption{
	(a) A graph with an exponential number of simple cycles, all of which can be found by starting from the edge $\mathbf{v_0 \rightarrow v_1}$.
	(b) The recursion tree of the Johnson algorithm for $\mathbf{n=6}$ constructed when the algorithm starts from the edge $\mathbf{v_0 \rightarrow v_1}$. Whereas a coarse-grained parallel algorithm explores the complete recursion tree using a single thread, our fine-grained parallel algorithms can explore different regions of the search tree in parallel using several threads. 
	}
	\vspace{-.15in}
	\label{fig:wc_example}
\end{figure}

\vspace{-.03in}

\begin{proposition}
\label{theorem:cg_workEff}
The coarse-grained parallel Johnson and Read-Tarjan algorithms are work efficient.
\end{proposition}

The proof of Proposition~\ref{theorem:cg_workEff} is trivial, and we omit it for brevity.

\begin{theorem}
\label{theorem:cg_scalability}
The coarse-grained parallel Johnson and Read-Tarjan algorithms are not scalable.
\end{theorem}

\begin{proof}
In this case, $T_{\infty}(n)$ represents the worst-case execution time of a search for cycles that starts from a single vertex or edge, and it depends on the number of cycles discovered during this search.
In the worst case, a single recursive search can discover all cycles of a graph.
An example of such graph is given in Figure~\ref{fig:wcEdgePar}, where each vertex $v_i$, with $i\in\{1,\ldots,n-1\}$, is connected to $v_0$ and to every vertex $v_j$ such that $j > i$.
In that graph, any subset of vertices $v_2,\ldots, v_{n-1}$ defines a different cycle.
Therefore, the total number of cycles in this graph is equal to the number of all such subsets $c = 2^{n-2}$.
Before the search for cycles, both the Johnson and the Read-Tarjan algorithm find all vertices that start a cycle, which is only $v_0$ in this case.
Therefore, the search for cycles will be performed only by one thread.
Because all cycles of the graph are discovered by a single thread, this thread performs all the work the sequential algorithm would perform, which leads to $T_{\infty}(n) = T_{1}(n)$.
Because it holds that $\lim\limits_{n\to\infty} T_{\infty}(n)/T_1(n) = 1$, the coarse-grained algorithms are not scalable based on Definition~\ref{def:scalability}.
\vspace{-.07in}
\end{proof}

\textbf{Summary.} Theorem~\ref{theorem:cg_scalability} shows that the main drawback of the coarse-grained parallel algorithms is their limited scalability.
This limitation is apparent for the graph shown in Figure~\ref{fig:wcEdgePar}, which has an exponential number of cycles in $n$. 
When using a coarse-grained parallel algorithm on this graph, all the cycles will be discovered by a single thread.
Because only one thread can be effectively utilised, increasing the number of threads will not result in a reduction of the overall execution time of the coarse-grained parallel algorithm. Figure~\ref{fig:motivfig} shows the load imbalance exhibited by the coarse-grained parallel algorithms in practice. Section~\ref{sect:experiments} demonstrates the limited scalability of coarse-grained parallel algorithms in further detail.
   
\section{Fine-grained parallel Johnson}
\label{sect:tpJohnson}

To address the load imbalance issues that manifest themselves in the coarse-grained parallel Johnson algorithm, we introduce the fine-grained parallel Johnson algorithm.
The main goal of our fine-grained algorithm is to enable several threads to explore a recursion tree concurrently as shown in Figure~\ref{fig:wc_exampleBT}, where each thread executes a subset of the recursive calls of this tree.
However, enabling several threads to explore the recursion tree concurrently is in conflict with the sequential depth-first exploration required by the Johnson algorithm to achieve a high pruning efficiency. 

\begin{figure}[t]
	\begin{subfigure}[t]{0.3\linewidth}
		\centering
        \includegraphics[height=5cm]{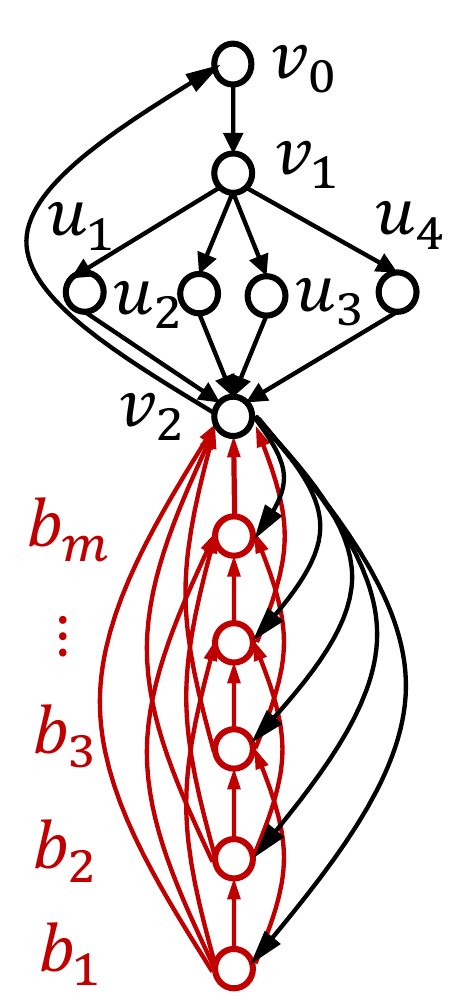}
	\vspace{-.05in}
		\caption{Example graph}
		\label{fig:fgj_exampleGraph}
	\end{subfigure}%
	\begin{subfigure}[t]{0.7\linewidth}
		\centering
        \includegraphics[height=5cm]{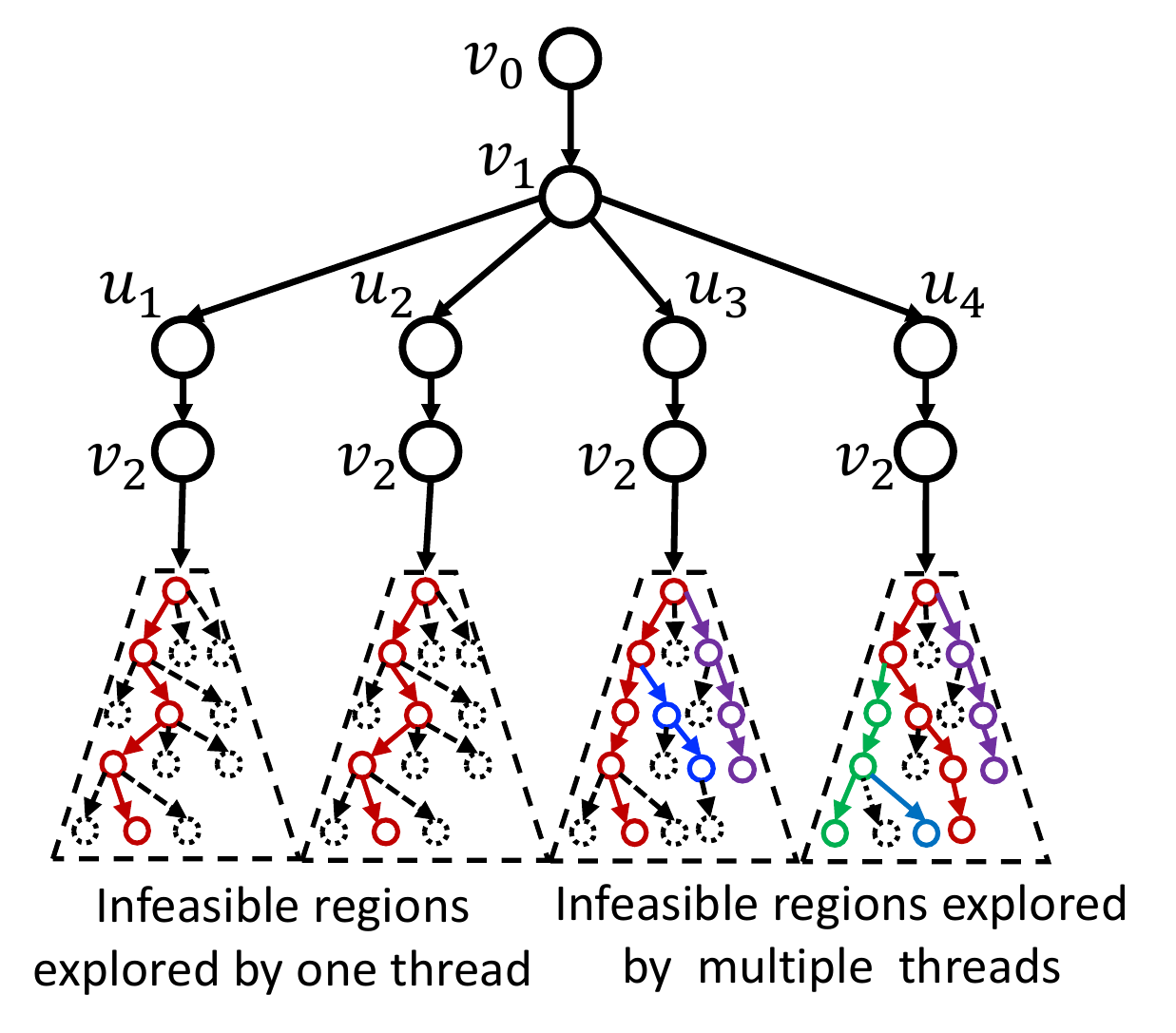}
	\vspace{-.05in}
		\caption{Recursion tree}
		\label{fig:fgj_exampleBT}
	\end{subfigure}
	\vspace{-.1in}
	\caption{
	(a) An example graph and (b) the recursion tree of our fine-grained parallel Johnson algorithm when searching for cycles that start from vertex $\mathbf{v_0}$, which
	can redundantly traverse the vertices $\mathbf{b_1,\ldots,b_m}$ several times.
	The serial Johnson algorithm would traverse these vertices only once.}
	\vspace{-.15in}
	\label{fig:fgj_example}
\end{figure}

\textbf{The na\"ive approach.} A straightforward way of enabling concurrent exploration of the recursion tree is to break the dependencies between different paths of the recursion tree by creating new copies of the $\mathit{Blk}$ and $\mathit{Blist}$ data structures when invoking child recursive calls.
In such a case, there would be no data sharing between different recursive calls, and different calls would store inconsistent copies of the data structures.
As a result, a recursive call would be unaware of the vertices visited and blocked by other calls that precede it in the depth-first order except for its direct ancestors in the recursion tree.
Hence, this approach exhaustively explores all maximal simple paths in the graph, and is identical to the brute-force solution of Tiernan (see Section~\ref{section:back_tiernan}).
When enumerating the simple cycles of the graph shown in Figure~\ref{fig:fgj_exampleGraph} starting from $v_0$, this approach would explore all $4 \times 2^{m-1}$ maximal simple paths instead of just four that would be visited by the Johnson algorithm.

\label{sect:tpj_copyOnSteal}

\textbf{Our approach.} To enable different threads to concurrently explore the recursion tree in a depth-first fashion while also taking advantage of the powerful pruning capabilities of the Johnson algorithm, each thread executing our fine-grained parallel Johnson algorithm maintains its own copy of the $\mathit{\Pi}$, $\mathit{Blk}$, and $\mathit{Blist}$ data structures.
Because a thread maintains a copy of the blocked vertex set $\mathit{Blk}$, it will not fruitlessly visit the vertices that it has already blocked.
Yet, different threads will store inconsistent copies of the $\mathit{Blk}$ and $\mathit{Blist}$ data structures, which could lead to some redundant work.
This redundant work could happen when different threads are exploring the same infeasible region as depicted in Figure~\ref{fig:fgj_exampleBT}.
However, because the threads executing our fine-grained parallel algorithm still take advantage of the powerful pruning methods of the Johnson algorithm, the amount of work performed will be significantly lower than that of the brute-force Tiernan algorithm.

\textbf{Copy-on-steal.} Our fine-grained parallel Johnson algorithm implements each recursive call as a separate task.
If a child task and its parent task are executed by the same thread, the child task reuses the $\mathit{\Pi}$, $\mathit{Blk}$, and $\mathit{Blist}$ data structures of the parent task.
However, if a child task has been stolen---i.e., it is executed by a thread other than the thread that created it, a new copy of these data structures are allocated by the child task.
In this way, each thread of our fine-grained parallel algorithm maintains its own copy of the $\mathit{\Pi}$, $\mathit{Blk}$, and $\mathit{Blist}$ data structures.
We refer to this mechanism as \emph{copy-on-steal}.

The problem with copying data structures between different threads upon task stealing is that the thread that has created the stolen task can modify its data structures before another thread steals this task. This inconsistency has to be somehow managed. 
A straightforward solution to this problem is to execute the stolen task after restoring the data structures to the state they were in when the stolen task was created. 
Even though such an approach would guarantee correct execution, the stealing thread would not be able to reuse the blocked vertices that have been discovered between the time the task was created and the time the task was stolen.
For example, in Figure~\ref{fig:cos_example}, assume that we are searching for simple cycles that start from $v_0$. While visiting $v_1$, the thread $T_1$ creates two new tasks, continues its depth-first search for simple cycles from vertex $w_1$ using the first task it has created, and pushes the second task it has created into its work queue to be continued from the vertex $u_1$.
Suppose that the thread $T_2$ steals this second task from $T_1$ while $T_1$ is visiting $w_3$.
At this point, $T_1$ has blocked the vertices $b_1$, $b_2$, $b_3$, and $b_4$ because it has discovered that it is not possible to construct a simple cycle that ends in $v_0$ while going through $b_1$, $b_2$, $b_3$, and $b_4$, once $v_1$ and $w_1$ have been visited.
If $T_2$ discards the changes $T_1$ has made to its data structures, it would still be able to discover the simple cycle that goes through $b_1$ and $b_2$ all the way to $v_0$.
However, $T_2$ would have to visit vertices $b_3$ and $b_4$ even though $T_1$ has already concluded that these vertices cannot lead to a simple cycle that ends in $v_0$, once $v_1$ has been visited.
Therefore, a na\"ive state restoration will lead to unnecessary work.

\begin{figure}[t]
	\centerline{
		\includegraphics[width=0.98\linewidth]{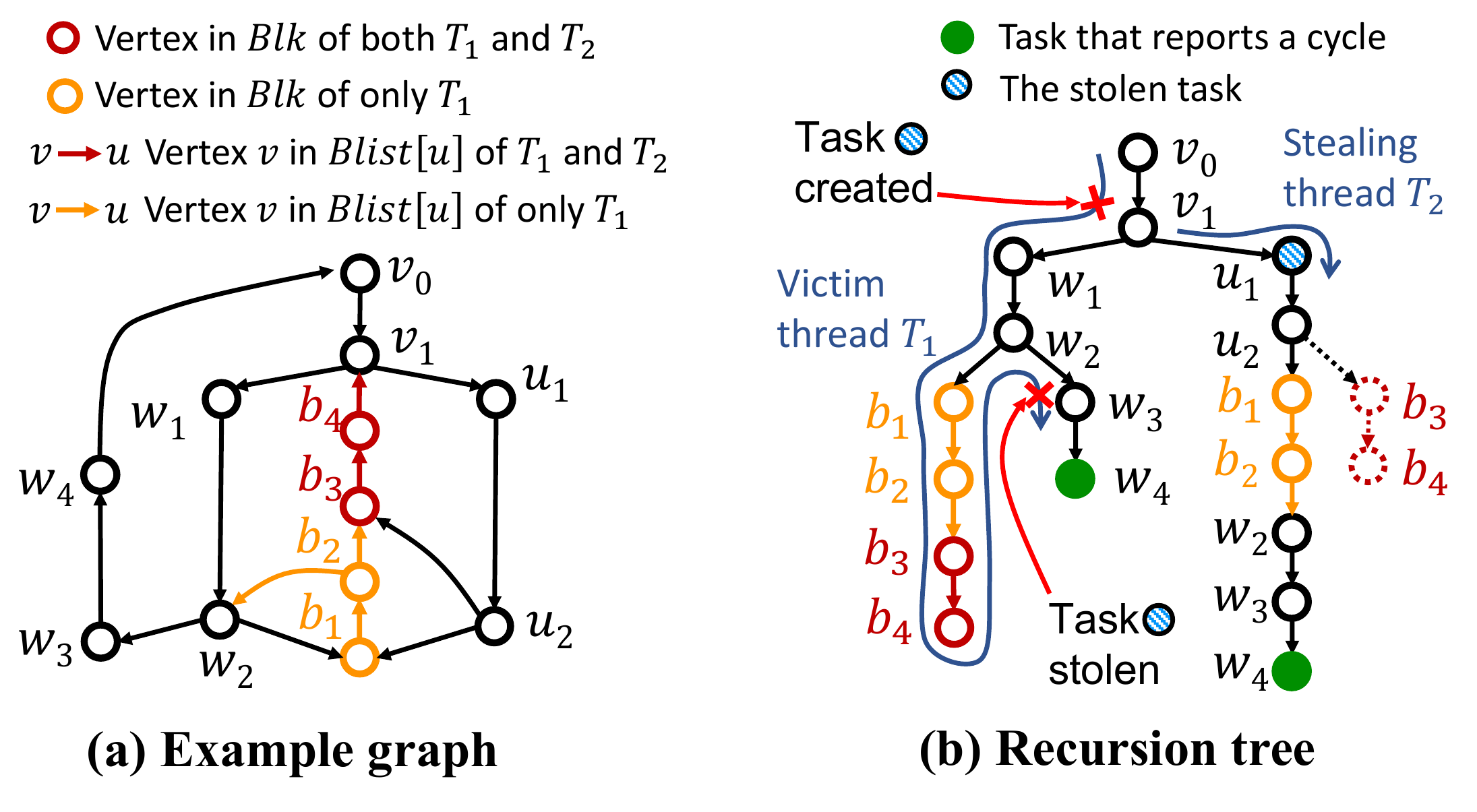}
	}
	\vspace{-.15in}
	\caption{
    (a) An example graph and (b) the recursion tree of our fine-grained parallel Johnson algorithm when searching for cycles that start from $\mathbf{v_0}$.
	The thread $\mathbf{T_2}$ can prune the dotted part of the tree by avoiding $\mathbf{b_3}$ and $\mathbf{b_4}$ that the thread $\mathbf{T_1}$ has blocked after creating the task stolen by $\mathbf{T_2}$.
	}
	\label{fig:cos_example}
	\vspace{-.1in}
\end{figure}

We have developed an alternative solution that makes it possible for the stealing thread to capitalise on the information already discovered by the thread from which it stole the task, henceforth referred to as the \emph{victim thread}.
The stealing thread first determines its initial path $\Pi_2$ by removing the vertices the victim thread has added to its path $\Pi_1$ since the stolen task was created.
The stealing thread then invokes the recursive unblocking procedure for each vertex $v \in \Pi_1 \setminus \Pi_2$, which removes these vertices from the set of blocked vertices $\mathit{Blk}$ and recursively unblocks the vertices from the blocked list $\mathit{Blist}[v]$ of each such $v$.
For instance, in Figure~\ref{fig:cos_example}, the stealing thread $T_2$ invokes recursive unblocking for vertices $w_3$, $w_2$, and $w_1$ because the victim thread $T_1$ has added them to its current path after the stolen task was created.
As a result of the recursive unblocking procedure, $T_2$ unblocks the vertices $b_1$ and $b_2$, but the vertices $b_3$ and $b_4$ remain blocked.
The vertices that remain blocked are the ones that cannot take part in a cycle that has $\Pi_2$ as a prefix because these vertices can only be unblocked by invoking the recursive unblocking procedure for the vertices in $\Pi_2$.
Therefore, some of the blocked vertices discovered by the victim thread can be avoided by the stealing thread, which would not have been possible if the stealing thread simply discarded the changes the victim thread has made to its $\mathit{Blk}$ and $\mathit{Blist}$ data structures. 

\textbf{Synchronization.} Without countermeasures, the copy-on-steal method can suffer from race conditions because the $\Pi$, $\mathit{Blk}$, and $\mathit{Blist}$ data structures can be concurrently accessed by a victim thread and several stealing threads.
For instance, if a stealing thread copies the data structures of a victim thread while the victim thread performs a recursive unblocking, the stealing thread could end up copying data that is not in a stable state.
To prevent such conditions, we define critical sections and implement coarse-grained locking by maintaining a mutex per thread.
The lock is acquired when the victim thread enters the recursive unblocking procedure or when the stealing threads attempt to copy the $\Pi$, $\mathit{Blk}$, and $\mathit{Blist}$ data structures from the victim thread. The lock is released when the recursive unblocking or the copy operation is complete.
In this way, the race conditions are eliminated and a correct execution is guaranteed. 

\textbf{Theoretical analysis.} We now show that the fine-grained parallel Johnson algorithm is not work efficient, but scalable.

\begin{theorem}
The fine-grained parallel Johnson algorithm is not work efficient.
\label{theorem:fgj_workEfficiency}
\vspace{-.07in}
\end{theorem}

\begin{proof}
According to Lemma~3 presented by Johnson~\cite{johnson_finding_1975}, a vertex cannot be unblocked more than once unless a cycle is found, and once a vertex is visited, it can be visited again only after being unblocked.
As a result, a vertex is visited and unblocked at most $c$ times by the Johnson algorithm.
In the fine-grained parallel Johnson algorithm executed using $p$ threads, each thread maintains a separate set of the data structures used for managing the blocked vertices ($\mathit{Blk}$ and $\mathit{Blist}$).
Because the threads are unaware of each-others blocked vertices, each vertex is visited at most $p c$ times, $c$ times by each thread.
Additionally, a vertex cannot be visited more than $s$ times because each maximal simple path of a graph is explored by a different thread in the worst case, and during each simple path exploration, a vertex is visited at most once. 
Therefore, the maximum number of times a vertex can be visited by the fine-grained parallel Johnson algorithm is $\min\left\{s, p c\right\}$.
When the Johnson algorithm visits a vertex, it also iterates through its outgoing edges, thus visiting all $n$ vertices executes in $O(n+e)$ time.
Prior to executing the recursive search, this algorithm checks if the input graph contains at least one cycle  using a single thread.
This check can be performed in $O(n+e)$ time.
As a result, the work performed by the fine-grained parallel Johnson algorithm is
\vspace{-.05in}
 \begin{equation}
 W_p(n) \in \begin{cases}
     O\left(n+e\right), & \text{if $c = 0$}, \\
     O\left(p c(n+e)\right), & \text{if $p < s/c$ and $c \neq 0$},\\
     O\left( s (n+e)\right), & \text{otherwise}.
   \end{cases}
 \end{equation} 
When $c > 0$, $p>1$, and $s>c$, the work performed by the fine-grained parallel Johnson algorithm $W_p(n)$ is greater than the execution time $T_1(n)$ of the sequential Johnson algorithm.  
Thus, the fine-grained parallel Johnson algorithm is not work efficient.
\vspace{-.05in}
\end{proof}

The work inefficiency of the fine-grained parallel Johnson algorithm occurs if more than one thread performs the work the sequential Johnson algorithm would perform between the discovery of two cycles.
We illustrate this behaviour using the graph from Figure~\ref{fig:fgj_exampleGraph}, which contains $c = 4$ cycles and $s = c \times 2^{m-1}$ maximal simple paths, each starting from vertex $v_0$.
When discovering each cycle, the fine-grained parallel Johnson algorithm explores an infeasible region of the recursion tree, as shown in Figure~\ref{fig:fgj_exampleBT}, in which vertices $b_1,\ldots,b_m$ are visited.
If this infeasible region is explored using a single thread, each vertex $b_i$, with $i\in \{1,\ldots,m\}$, will be visited exactly once.
However, if $p$ threads are exploring the same infeasible region of the recursion tree, vertices $b_1,\ldots,b_m$ will be visited up to $p$ times because the threads are unaware of each-others blocked vertices.
In this case, the fine-grained parallel Johnson algorithm performs more work than necessary, and, thus, it is not work efficient.
Additionally, each infeasible region of the recursion tree that visits vertices $b_1,\ldots,b_k$ can be executed by at most $s/c = 2^{m-1}$ threads because there are $2^{m-1}$ maximal simple paths that can be explored in each infeasible region.
In this case, each vertex $b_i$, with $i\in \{1,\ldots,m\}$, is visited up to $s$ times, and, thus, the fine-grained parallel Johnson algorithm behaves as the Tiernan algorithm (see Section~\ref{section:back_tiernan}).

\begin{lemma}
The depth $T_{\infty}(n)$ of the fine-grained parallel Johnson algorithm is in $O(n+e)$.
\vspace{-.07in}
\label{lemma:fgJ_tinf}
\end{lemma}

\begin{proof}
The worst-case depth of this algorithm occurs when a thread performs copy-on-steal and explores a path of length $n$.
Performing copy-on-steal requires $O(n+e)$ operations because at most $n$ vertices in $\Pi$ and $\mathit{Blk}$, and at most $e$ pairs of vertices in $\mathit{Blist}$ are accessed during the copy-on-steal.
Exploring the path of length $n$ requires $O(n+e)$ operations because a recursive call that visits a vertex $v$ of this path also iterates through every outgoing edge of $v$.
As a result, the depth of this algorithm is $T_{\infty}(n) \in O(n+e)$.
\vspace{-.05in}
\end{proof}

\begin{theorem}
\label{theorem:fgJohnScal}
The fine-grained parallel Johnson algorithm is scalable when $\lim\limits_{n \to \infty} c = \infty$.
\vspace{-.08in}
\end{theorem}

\begin{proof}
For this algorithm, $T_1(n) \in O((n+e)(c+1))$ and $T_{\infty}(n) \in O(n+e)$ (see Lemma~\ref{lemma:fgJ_tinf}).
Given our assumption that $\lim\limits_{n \to \infty} c = \infty$, we have $\lim\limits_{n\to\infty}\dfrac{T_{\infty}(n)}{T_1(n)}  =  \lim\limits_{n\to\infty} \dfrac{n+e}{(n+e)(c+1)} = 0$.
Therefore, this algorithm is scalable based on Definition~\ref{def:scalability}.
Note that it is sufficient for $c$ to increase sublinearly in $n$ for this proof to hold.
\vspace{-.07in}
\end{proof}

Even though the fine-grained parallel Johnson algorithm is scalable, a strong or weak scalability is not guaranteed due to the work inefficiency of this algorithm.
Nevertheless, our experiments show that this algorithm is strongly scalable in practice (see Figure~\ref{fig:scalfig}).

\textbf{Summary.} Our relaxation of the strictly depth-first-search-based recursion-tree exploration reduces the pruning efficiency of the Johnson algorithm. 
In the worst case, the fine-grained parallel Johnson algorithm could perform as much work as the brute-force Tiernan algorithm does---i.e., $O(s(n+e))$. However, in practice this worst-case scenario does not happen (see Section~\ref{sect:experiments}).
In addition, our fine-grained parallel Johnson algorithm can suffer from synchronisation issues in some rare cases (see Section~\ref{sect:experiments}) because our copy-on-steal mechanism can lead to long critical sections. In the next section, we introduce a fine-grained parallel algorithm that is scalable, work efficient, and less prone to synchronisation issues.

\section{Fine-grained parallel Read-Tarjan}
\label{sect:tpReadTarjan}

In this section, we show that the Read-Tarjan algorithm is straightforward to parallelise in a scalable and work efficient way.
Because the Read-Tarjan algorithm allocates a new copy of the $Blk$ set during each path extension computation, a recursive call can compute different path extensions in an arbitrary order.
Additionally, discovery of a new path extension results in the invocation of a single recursive call, and these calls can be executed in an arbitrary order.
Consequently, several threads can concurrently explore different paths of the same recursion tree constructed by the Read-Tarjan algorithm for a given starting edge. There are no data dependencies or ordering requirements between different calls apart from those that exist between a parent and a child.
To exploit the parallelism available during the recursion tree exploration, we execute each recursive call and each path extension computation as a separate task that can be independently scheduled. We refer to the resulting algorithm as the fine-grained parallel Read-Tarjan algorithm.

To prevent different threads from concurrently modifying the current path $\Pi$ being explored, each task receives a copy of $\Pi$ from its parent task.
To improve the pruning efficiency, each task also receives a copy of the blocked vertex set $\mathit{Blk}$ from its parent task.
However, unlike the Johnson algorithm, 
the fine-grained parallel Read-Tarjan algorithm does not communicate the $\mathit{Blk}$ sets from child tasks back to their parent tasks.
This parallel algorithm also takes advantage of the copy-on-steal mechanism (see Section~\ref{sect:tpj_copyOnSteal}) to eliminate unnecessary copy operations between tasks executed by the same thread.
However, because the Read-Tarjan algorithm does not use the $\mathit{Blist}$ data structures and the recursive unblocking procedure used by the Johnson algorithm, the critical sections of the fine-grained parallel Read-Tarjan algorithm is much shorter than those of the fine-grained parallel Johnson algorithm.

\textbf{Theoretical analysis.} We now show that the fine-grained parallel Read-Tarjan algorithm is both work efficient and scalable.

\begin{theorem}
The fine-grained parallel Read-Tarjan algorithm is work efficient.
\vspace{-.03in}
\label{theorem:fgRT_workEfficiency}
\end{theorem}

\begin{proof}
Because the Read-Tarjan algorithm executes $O(c)$ recursive calls~\cite{read_bounds_1975}, and each path extension exploration invokes at most one recursive call, our fine-grained parallel Read-Tarjan algorithm is executed using $O(c)$ tasks.
Each task executes a DFS that explores at most $n$ vertices and $e$ edges.
Additionally, each task receives a copy of $\Pi$ and $\mathit{Blk}$, and because these data structures contain at most $n$ vertices, the overhead of copying them is $O(n)$.
Therefore, this algorithm performs $O(n+e)$ work per task.
Because the same amount of work is performed even if there are no cycles in the graph, the total amount of work this algorithm performs is $W_p(n) = O\left((n+e)(c+1)\right)$.
As a result, based on Definition~\ref{def:workEfficiency}, the fine-grained parallel Read-Tarjan algorithm is work efficient.
\vspace{-.07in}
\end{proof}

Using the graph from Figure~\ref{fig:fgj_exampleGraph}, threads of the fine-grained parallel Read-Tarjan algorithm that start from $v_0$ independently explore four different path extensions $\Pi_{E} = v_{1} \rightarrow u_{i} \rightarrow v_2 \rightarrow v_0$, with $i\in \{1\ldots4\}$.
When exploring a path extension $\Pi_{E}$, each thread invokes a DFS starting from $v_2$ to explore a different infeasible region of the search tree as shown in Figure~\ref{fig:fgj_exampleBT}.
Because the DFS would fail to find any other path extensions, the same infeasible region will not be explored more than once.
Therefore, the amount of work the fine-grained parallel Read-Tarjan algorithm performs does not increase compared to its single-threaded execution.

\begin{lemma}
The depth $T_{\infty}(n)$ of the fine-grained parallel Read-Tarjan algorithm is in $O(n(n+e))$.
\vspace{-.07in}
\label{lemma:fgRT_tinf}
\end{lemma}

\begin{proof}
In the worst case, a thread executing this algorithm invokes a recursive call for each vertex of its longest simple cycle, which has a length of at most $n$.
Each recursive call executes a DFS that can visit $n$ vertices and $e$ edges of the graph.
Therefore, the depth of this algorithm is $T_{\infty}(n) \in O\left(n(n+e)\right)$.
\vspace{-.05in}
\end{proof}

The worst-case depth of the fine-grained parallel Read-Tarjan algorithm can be observed when this algorithm is executed on the graph given in Figure~\ref{fig:wcEdgePar}. This graph has $c = 2^{n-2}$ cycles and the length of its longest cycle $v_0 \rightarrow v_1 \ldots \rightarrow v_{n-1} \rightarrow v_0$ is $n$.
The algorithm invokes a recursive call for each vertex of the cycle and performs a DFS in each such call, which leads to $T_{\infty} \in O(n(n+e))$.

\begin{theorem}
The fine-grained parallel Read-Tarjan algorithm is strongly scalable when $\lim\limits_{n \to \infty} c/n = \infty$.
\vspace{-.07in}
\label{theorem:fgRT_scalability}
\end{theorem}

\begin{proof}
Because the fine-grained parallel Read-Tarjan algorithm is work-efficient, we can apply Brent's rule~\cite{brent_parallel_1974} as follows:
\[\dfrac{T_1(n)}{p} \leq T_p(n) \leq \dfrac{T_1(n)}{p} + T_{\infty}(n).\]
Substituting $T_1(n)$ with $O((c+1)(n+e))$ and $T_{\infty}(n)$ with $O(n(n+e))$ (see Lemma~\ref{lemma:fgRT_tinf}), for a positive constant $C_0$, it holds that
\[\slfrac{1}{\left(\dfrac{1}{p} + \dfrac{T_{\infty}(n)}{T_1(n)}\right)} = \slfrac{1}{\left(\dfrac{1}{p} + \mathit{C_0}\dfrac{n}{c+1}\right)} \leq \dfrac{T_1(n)}{T_{\infty}(n)} \leq p.\]
Given that $\lim\limits_{n \to \infty} c/n = \infty$, there exist $n_0>0, C_1>0$ such that if $n > n_0$, then $(c+1)/n > C_1 p$.
Thus, for every $n > n_0$, it holds that $k p \leq \tfrac{T_1(n)}{T_{\infty}(n)} \leq p$, where $k = C_1/(C_0+C_1) < 1$.
As a result, $\tfrac{T_1(n)}{T_{\infty}(n)}  = \Theta(p)$, which, based on Definition~\ref{def:strongScalability}, completes the proof.
Note that this proof requires $c$ to grow superlinearly with $n$.
\vspace{-.05in}
\end{proof}

\textbf{Summary.} The pruning efficiency of the Read-Tarjan algorithm is not affected by the fine-grained parallelisation. The fine-grained parallel Read-Tarjan algorithm performs $O((n+e)(c+1))$ work: the same as the work performed by its serial version.
In addition, the synchronization overheads of the fine-grained parallel Read-Tarjan algorithm are not as significant as those of the fine-grained Johnson algorithm because of its shorter critical sections.
Furthermore, this algorithm is the only asymptotically-optimal parallel algorithm for cycle enumeration for which we have proved strong scalability. 

\section{Extensions to temporal cycles}
\label{sect:temporalCycle}

To efficiently enumerate temporal cycles, we take advantage of the 2SCENT algorithm contributed by Kumar and Calders~\cite{kumar_2scent_2018}, which is based on the Johnson algorithm. The 2SCENT algorithm introduces two highly effective optimisations, called \textit{closing times} and \textit{path bundling}.
The closing times optimisation extends the concept of the blocked vertex set $\mathit{Blk}$ to keep track of the blocked temporal edges.
The path bundles optimisation enables a single search to simultaneously explore several temporal paths that share a common sequence of vertices.
To enable efficient temporal cycle enumeration, we have incorporated both optimisations into our coarse- and fine-grained parallel algorithms introduced in the prior sections.
Given that our parallel algorithms for temporal cycle enumeration are based on our parallel formulations of the Johnson and the Read-Tarjan algorithms, our conclusions regarding the work efficiency and scalability, summarised in Table~\ref{tab:theoSummary}, remain valid. 
However, we omit the respective proofs due to space constraints.

The 2SCENT algorithm also uses a preprocessing step that reduces the number of vertices visited during its search for cycles.
However, this preprocessing step has a strictly sequential formulation because it processes the edges in the increasing order of their timestamps.
Moreover, the time complexity of the preprocessing step of the 2SCENT algorithm is in the order of the time complexity of its recursive search for cycles.
In our own implementation, we use a lighter-weight linear-time preprocessing algorithm, inspired by the algorithms for computing strongly-connected components~\cite{Tarjan1972DepthFirstSA, fleischer_identifying_2000}, which can be parallelised in a scalable manner.

\textbf{Our scalable preprocessing method} computes a \textit{cycle-union} for each starting edge, which is the set of vertices that take part in one or more temporal cycles starting from that edge.
The cycle-union of a given starting edge $v_0 \rightarrow v_1$ is computed as the intersection between the set of vertices reachable from vertex $v_1$ and the set of vertices from which vertex $v_0$ is reachable.
When computing temporal cycles, we say that a vertex $u$ is reachable from a vertex $v$ if there exists a simple path from $v$ to $u$, in which the edges appear in the increasing order of their timestamps.
When performing the reachability analysis under time window constraints, we only consider the paths that belong to a time window of a given size $\delta$.
This preprocessing method is lightweight because each cycle-union can be computed in $O(n+e)$ time, similarly to the computation of a strongly-connected component~\cite{fleischer_identifying_2000} of a graph, and it is also straightforward to parallelise because the cycle-unions can be computed independently for each starting edge or starting vertex.

\section{Experimental evaluation}
\label{sect:experiments}

This section evaluates the performance of our coarse- and fine-grained parallel versions of the Johnson and the Read-Tarjan algorithms on temporal graphs.
As Table~\ref{tab:relWork} shows, we are the only ones to offer fine-grained parallel versions of the state-of-the-art algorithms by Johnson and Read-Tarjan. However, all the methods covered in Table~\ref{tab:relWork} can be parallelised using the coarse-grained approach we described in Section~\ref{sect:vertEdgePar}, which we use as our main comparison point. Furthermore, we provide direct comparisons with 2SCENT~\cite{kumar_2scent_2018} because it is the only other work that supports time window constraints and can perform temporal cycle enumeration.

Because exhaustive enumeration of simple cycles is not tractable in general, it is common to search for cycles under some constraints (see Table~\ref{tab:relWork}).
In the experiments, we use time-window constraints when searching for both simple and temporal cycles.
Our experiments are performed using the temporal graphs listed in Table~\ref{tab:dataset}.
The TR, FR, and MS graphs are from \emph{Harvard Dataverse}~\cite{jankowski_spreading_2017}, the NL graph is from \emph{Konect}~\cite{kunegis_konect_2013}, the AML graph is from the \emph{AML-Data} repository~\cite{amldata}, and the rest are from \emph{SNAP}~\cite{snapnets}.
We control the complexity of cycle enumeration by selecting the time-window size $\delta$ appropriately for each graph.
The window sizes used in our experiments are given in Table~\ref{tab:dataset}.
We do not report simple cycle enumeration results for the MS graph because, in this case, our algorithms did not finish in $12h$ even if we set $\delta = 1s$.
Note that we use larger time windows when enumerating temporal cycles because the complexity of enumerating temporal cycles is lower.

The experiments are performed on a cluster of four Intel Xeon Phi 7210 - Knights Landing (KNL) processors~\cite{sodani_knights_2015}.\footnote{Intel and Intel Xeon are trademarks or registered trademarks of Intel Corporation or its subsidiaries in the United States and other countries.}
An Intel KNL CPU has $64$ physical cores and supports $256$ simultaneous threads, which makes it an ideal platform for evaluating the scalability of parallel algorithms.
This cluster enables execution of $1024$ simultaneous threads on $256$ physical CPU cores.
We use the \textit{Threading Building Blocks} (TBB)~\cite{kukanov_foundations_2007} library for parallelising the algorithms on a single processor, and we distribute the execution of the algorithms across multiple processors using the Message Passing Interface (MPI)~\cite{mpi_1993}.
When using distributed execution, each processor stores a copy of the input graph in its main memory and searches for cycles starting from a different set of graph edges.
The starting edges are divided among the processors such that when the edges are ordered in the ascending order of their timestamps, $k$ consecutive edges in that order are assigned to $k$ different processors.
Each processor then uses its own dynamic scheduler to balance the workload of the recursive searches that start from its given set of starting edges. 

\begin{table}[t]
\centering
\small 
\caption{Temporal graphs. Time span $\mathbf{T}$ is in days.
Figure~\ref{fig:main_comparison_simple} and~\ref{fig:main_comparison_temp}
use the time window sizes $\mathbf{\delta_{s}}$~and~$\mathbf{\delta_{t}}$, respectively.
}
\vspace{-.1in}
\begin{tabular}{lccc|cc}
	\textbf{Graph}  &  $\mathbf{n}$ & $\mathbf{e}$ & $\mathbf{T}$ & $\mathbf{\delta_{s}}$& $\mathbf{\delta_{t}}$\\ \hline
    \textbf{bitcoinalpha (BA)}  & 3.3 k & 24 k & 1901 & 71h & 3000h\\
    \textbf{bitcoinotc (BO)}   &  4.8 k  & 36 k & 1903 & 75h & 1000h\\
     \textbf{CollegeMsg (CO)}     &  1.3 k  & 60 k & 193 & 3h & 96h\\
	\textbf{email-Eu-core (EM)} &      824 &   332 k  & 803 & 4h & 144h\\
	\textbf{mathoverflow (MO)}   &      16 k & 390 k & 2350 & 30h & 288h\\
    \textbf{transactions (TR)} & 83 k & 530 k & 1803 & 72h & 800h\\
	\textbf{higgs-activity (HG)}  &      278 k & 555 k & 6 & 3000s & 72h\\
  \textbf{askubuntu (AU)} &  102 k & 727 k & 2613 & 20h & 336h\\
	\textbf{superuser (SU)}     &    138 k &  1.1 M & 2773 & 5h & 168h\\
	\textbf{wiki-talk (WT)} &       140 k & 6.1 M & 2277 & 12h & 144h\\
     \textbf{friends2008 (FR)}   &  481 k   & 12 M  & 1826 & 1300s & 5h\\ 
     \textbf{wiki-dynamic (NL)}&  1 M   & 20 M  & 3602 & 29s & 1000s\\ 
     \textbf{messages (MS)}   &  313 k   & 26 M  & 1880 & / & 4h\\ 
     \textbf{AML-Data (AML)}  &  10 M   & 34 M  & 30 & 48h & 720h \\ 
	\textbf{stackoverflow (SO)} &    2.0 M &   48 M &  2774 & 3h & 66h\\
    \hline
\end{tabular}
\label{tab:dataset}
\vspace{-.15in}
\end{table}

\begin{figure*}[ht!]
		\centering
	\begin{subfigure}[]{1\textwidth}
		\includegraphics[width=1\linewidth]{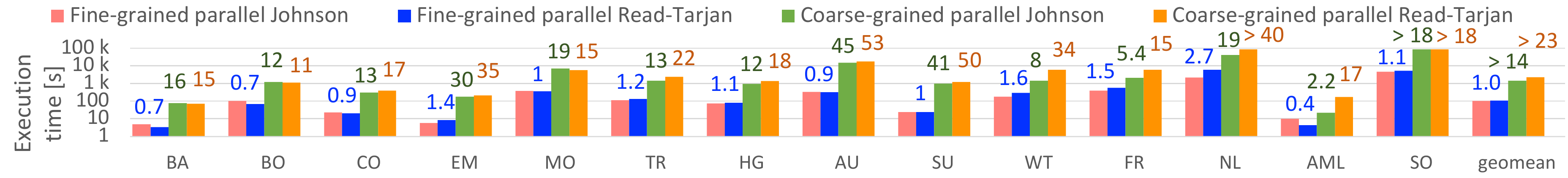}
	\vspace{-.25in}
		\caption{Performance of the parallel algorithms for finding all simple cycles within a time window of size $\delta_s$ given in Table~\ref{tab:dataset}.
		}
		\label{fig:main_comparison_simple}
	\end{subfigure}
	\\
	\begin{subfigure}[]{1\textwidth}
		\includegraphics[width=1\linewidth]{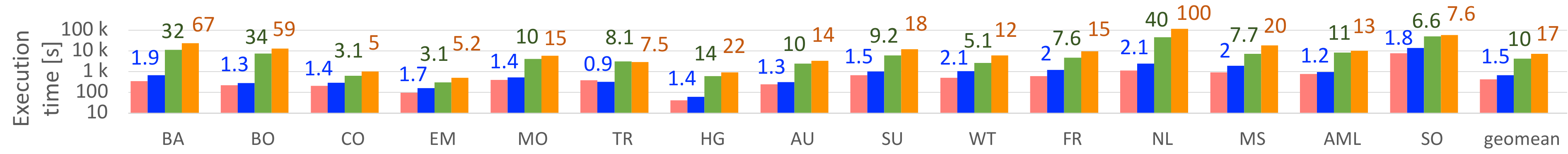}
	\vspace{-.25in}
		\caption{Performance of the parallel algorithms for finding all temporal cycles within a time window of size $\delta_t$ given in Table~\ref{tab:dataset}.}
		\label{fig:main_comparison_temp}
	\end{subfigure}
	\vspace{-.1in}
	\caption{Comparisons between the fine-grained and the coarse-grained parallel versions of the Johnson and the Read-Tarjan algorithms for (a) simple and (b) temporal cycle enumeration using $\mathbf{1024}$ threads.
	The numbers above the bars show the execution time of each algorithm relative to that of the fine-grained parallel Johnson algorithm for the same benchmark.
	}
	\vspace{-.05in}
	\label{fig:main_comparison}
\end{figure*}

\begin{figure*}[t]
	\centerline{
		\includegraphics[width=1\linewidth]{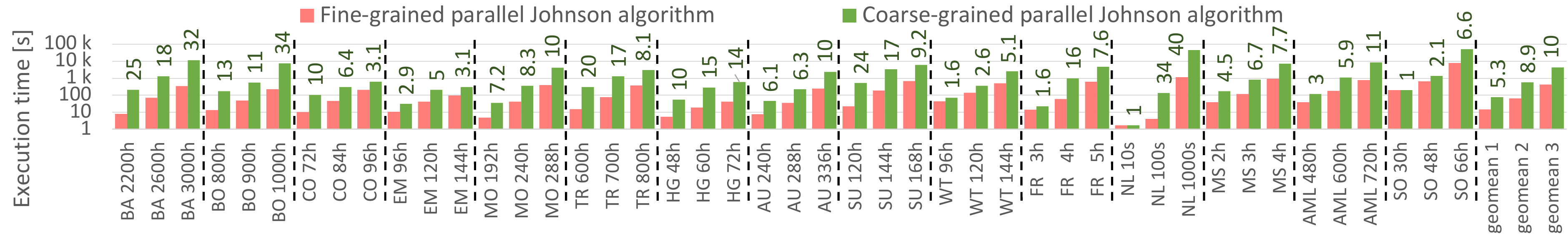}
	}
	\vspace{-.15in}
	\caption{
	The speed-up of the fine-grained parallel Johnson algorithm with respect to the coarse-grained parallel Johnson algorithm for different time window sizes. Larger time windows increase the performance gap between the two algorithms.
	}
	\label{fig:tw-effect}
	\vspace{-.15in}
\end{figure*}

\textbf{The granularity of the tasks} has a significant impact on the performance of the parallel cycle enumeration algorithms. Figure~\ref{fig:main_comparison} compares the coarse- and fine-grained parallel versions of the Johnson and the Read-Tarjan algorithms.
These comparisons are provided for both simple cycle enumeration and temporal cycle enumeration, respectively in Figures~\ref{fig:main_comparison_simple} and~\ref{fig:main_comparison_temp}. 
We observe that our fine-grained parallel algorithms outperform the coarse-grained parallel algorithms by an order of magnitude both for simple cycle enumeration and for temporal cycle enumeration.
This behavior is a clear outcome of the scalability of our fine-grained parallelisation.

Figure~\ref{fig:tw-effect} shows the impact of the time window size on the fine-grained and coarse-grained parallel Johnson algorithms when performing temporal cycle enumeration.
Note that enumerating cycles in larger time windows is more challenging because larger time windows contain a larger number of cycles. Interestingly, increasing the size of the time window increases the performance gap between the fine-grained and the coarse-grained Johnson algorithms.

The work performed by the algorithms evaluated can be quantified based on the number of edges visited during their execution. Based on this metric, our fine-grained parallel Johnson algorithm on average performs $6.1\%$ more work than the work-efficient coarse-grained parallel Johnson algorithm does when enumerating simple cycles. The maximum difference observed is around $14\%$. The difference is always less than $1\%$ when enumerating temporal cycles.

\textbf{The evaluation of the scalability} of the coarse-grained parallel and fine-grained parallel algorithms is performed in Figure~\ref{fig:scalfig}.
All three algorithms evaluated perform temporal cycle enumeration and use up to $1024$ software threads. We also report the performance of the 2SCENT algorithm~\cite{kumar_2scent_2018}.
In this setting, the only difference between our single-threaded Johnson algorithm and 2SCENT is the scalable pre-processing method we introduced in~Section~\ref{sect:temporalCycle}.

\begin{figure*}[t]
	\centerline{
		\includegraphics[width=1\linewidth]{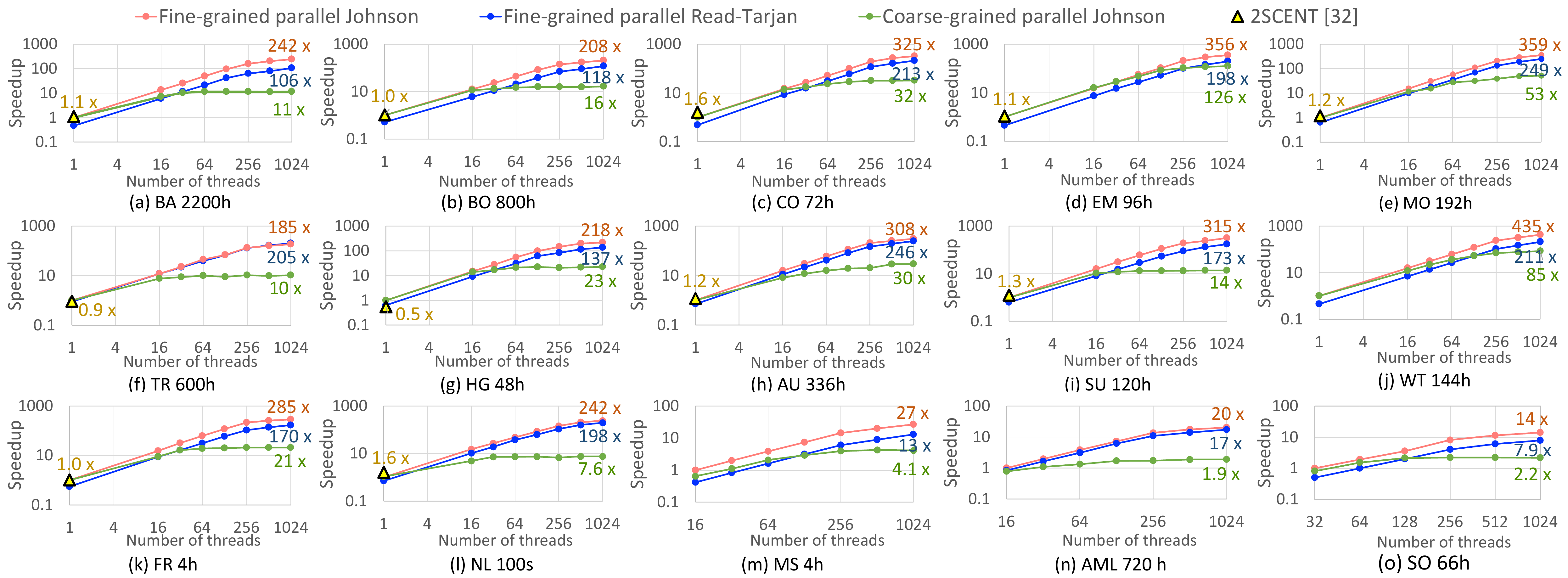}
	}
	\vspace{-.15in}
	\caption{Effect of the number of threads on the performance of temporal cycle enumeration algorithms. The baseline is our fine-grained parallel Johnson algorithm.
	The relative performance of 2SCENT~\cite{kumar_2scent_2018} is shown when it completes in $\mathbf{24h}$. Note that the 2SCENT implementation is single-threaded and the single-threaded execution results are not available for all graphs. 
	}
	\label{fig:scalfig}
	\vspace{-.1in}
\end{figure*}

The performance of the fine-grained parallel Johnson and Read-Tarjan algorithms improves linearly until $256$ threads. When we use more than $256$ threads, the CPU cores start performing simultaneous multithreading, which leads to a sublinear performance scaling.
In case of the WT graph, our fine-grained parallel versions of the Johnson and the Read-Tarjan algorithms are respectively $435\times$ and $470\times$ faster than their serial versions  when we use $1024$ threads. 
In addition, our fine-grained parallel Johnson algorithm is on average  $260\times$ faster than 2SCENT when 2SCENT completes in $24$ hours.

\textbf{The Johnson and the Read-Tarjan algorithms} have comparable performances as shown in Figure~\ref{fig:main_comparison}.
However, our fine-grained parallel Read-Tarjan algorithm is slightly slower than our fine-grained parallel Johnson algorithm.
This behaviour is expected given that the Read-Tarjan algorithm performs more edge visits than the Johnson algorithm despite having the same worst-case time complexity (see Section~\ref{sect:read_tarjan}). In our experiments, the fine-grained parallel Read-Tarjan algorithm on average performs $47\%$ more edge visits than the fine-grained parallel Johnson algorithm. 

The simple cycle enumeration results for AML are clear outliers. The coarse-grained parallel Read-Tarjan algorithm performs $2.4\times$ more edge visits than the coarse-grained parallel Johnson algorithm, yet it is $7.7\times$ slower due to a more severe load imbalance. However, our fine-grained parallel Johnson algorithm is $2.3\times$ slower than our fine-grained  parallel Read-Tarjan algorithm.  In this case, the fine-grained parallel Johnson algorithm is only $37\times$ faster than its serial version. However, the fine-grained parallel Read-Tarjan algorithm exhibits a good scaling and is $214\times$ faster than its serial version.

Our analysis has shown that such a role reversal is not caused by the work inefficiency of our fine-grained parallel Johnson algorithm, which performs only $10\%$ more edge visits than its serial version when enumerating the simple cycles of AML. 
The reason is the synchronization overheads exerted on our fine-grained parallel Johnson algorithm by recursive unblocking (see Section~\ref{sect:tpj_copyOnSteal}).
In fact, the synchronization overheads of our fine-grained parallel Johnson algorithm are visible only when enumerating the simple cycles of AML, which can be explained by a very low cycle-to-vertex ratio observed in this case.
Because a vertex is blocked if it cannot take part in a cycle, the probability of a vertex being blocked is higher when the cycle-to-vertex ratio is lower.
In consequence, more vertices are unblocked during the recursive unblocking of the fine-grained parallel Johnson algorithm, which leads to longer critical sections and more contention on the locks. Nevertheless, our fine-grained parallel Johnson algorithm achieves a good trade-off between pruning efficiency and lock contention in most cases.

\section{Conclusions}
\label{sect:conclusion}

This work has made three contributions to the area of parallel cycle enumeration.
First, we have introduced fine-grained parallel versions of the Johnson and the Read-Tarjan algorithms for enumerating simple and temporal cycles.
We have shown that our fine-grained parallel algorithms are scalable both in theory and in practice.
We have evaluated our algorithms on 15 temporal graph datasets, and demonstrated a near-linear performance scaling on a compute cluster with a total number of $256$ CPU cores that can execute $1024$ simultaneous threads, where
our parallel algorithms achieved an up to $470\times$ speedup with respect to their serial versions.

Secondly, we have shown that the coarse-grained parallel versions of the Johnson and the Read-Tarjan algorithms are not scalable.
When using $1024$ simultaneous software threads, our fine-grained parallel algorithms are on average an order of magnitude faster than the coarse-grained parallel algorithms.
In addition, the performance gap between the fine-grained and coarse-grained parallel algorithms increases as we use more physical cores. The performance gap increases as we increase the time window size as well.

Thirdly, we have shown that our fine-grained parallel Johnson algorithm is not work efficient.
Yet, it outperforms our fine-grained parallel Read-Tarjan algorithm in most of our experiments.
In some rare cases, our fine-grained parallel Johnson algorithm can suffer from synchronisation overheads. In such cases, our fine-grained parallel Read-Tarjan algorithm offers a more scalable alternative.

\begin{acks}
The support of Swiss National Science Foundation (project number 172610) for this work is gratefully acknowledged.
The authors would like to thank Haris Pozidis and Radu Stoica from IBM Research Europe - Zurich for their valuable comments on this work.
\end{acks}

\balance

\bibliographystyle{ACM-Reference-Format}
\bibliography{References}

\end{document}